\documentclass[a4paper,11pt]{article}
\usepackage[margin=27mm]{geometry}
\usepackage[T1]{fontenc}
\usepackage{lmodern}
\usepackage{amsmath,amssymb,amsthm,mathtools,microtype}
\usepackage[numbers,sort&compress]{natbib}
\usepackage{enumitem}
\setlist[enumerate]{itemsep=4pt,topsep=6pt}
\usepackage{fancyhdr}
\usepackage[colorlinks=true,linkcolor=blue,citecolor=blue,urlcolor=blue]{hyperref}
\newtheorem{theorem}{Theorem}[section]
\newtheorem{lemma}[theorem]{Lemma}
\newtheorem{proposition}[theorem]{Proposition}

\newtheorem{example}[theorem]{Example}
\theoremstyle{definition}\newtheorem{definition}[theorem]{Definition}
\theoremstyle{remark}

\newcommand{\ii}{\mathrm i}
\newcommand{\EE}{\mathbb E}
\newcommand{\PP}{\mathbb P}
\newcommand{\CC}{\mathbb C}
\newcommand{\RR}{\mathbb R}
\newcommand{\ZZ}{\mathbb Z}
\newcommand{\NN}{\mathbb N}

\newcommand{\Res}{\operatorname{Res}}
\newcommand{\tr}{\operatorname{tr}}
\newcommand{\ran}{\operatorname{ran}}
\title{Integer Sampling and the Identifiability of Mellin Residues}
\author{Yifan He\\[0.3em]
\small School of Physical Science and Technology, Lanzhou University\\
\small Lanzhou 730000, China\\
\small\texttt{heyf21@lzu.edu.cn}}

\date{}

\begin{document}
\maketitle
\begin{abstract}
I prove that all Mellin residues are determined exactly by integer samples
known up to errors smaller than every algebraic order, for an explicit
class of functions consisting of a slow component with a complete
inverse-power expansion and countably many carriers
$\exp(2\pi i\sigma 2^k a^x)$ with admissible envelopes.
A single exceptional null set of parameters $a>1$
works for the entire function class. The key estimate is an eventual uniform
lower bound for a polynomially weighted sampling Gram matrix on windows of
length $N^{1/2}$, with $O(N^{1/4})$ active modes. I derive this bound from
parameter-averaged oscillatory integrals and a finite trace-moment comparison
with product Haar measure, retaining the exact resonances among dyadic
frequencies. Summed Taylor estimates then identify every coefficient of the
slow asymptotic expansion, and a Mellin continuation argument identifies
these coefficients with the actual residues. I also give an alternative
derivation from existing higher-correlation estimates, an infinite-mode
example, and exact counterexamples at integer parameters. The result
supplies an interpolation-invariance statement for a model class motivated
by Mellin averaging in quantum gravity.
\end{abstract}

\section{Introduction}\label{sec:intro}

\subsection{Wormholes, ensembles, and asymptotic averaging}
Gravitational path integrals with several asymptotic boundaries can receive connected contributions from wormhole geometries. Such contributions motivate a statistical interpretation of quantities that, in a fixed microscopic theory, are definite observables. The issue is especially sharp when the boundary description is a specified quantum system rather than an explicitly chosen ensemble. It has led to several related, but mathematically distinct, ways of connecting gravitational expansions with statistical averages.

A concrete setting is Jackiw--Teitelboim gravity. Saad, Shenker and Stanford~\cite{SSSMatrix} identify its partition functions on surfaces of arbitrary genus and with multiple boundaries with the genus expansion of a matrix integral. This realizes an ensemble description at the level of that expansion; its nonperturbative completion is not unique. Marolf and Maxfield~\cite{MarolfMaxfield} examine the relation between spacetime wormholes, baby-universe states and ensembles, emphasizing the role of null states in the gravitational inner product. These constructions clarify why ensemble-like structures can arise, while also exposing the distinction between an ensemble description and a particular member of an ensemble.

The spectral form factor gives another useful perspective. For a
Hamiltonian $H_N$ for which $\exp(-\beta H_N)$ is trace class, one considers
\[
Z_N(\beta+it)=\operatorname{Tr}\exp[-(\beta+it)H_N],
\qquad
K_N(\beta,t)=|Z_N(\beta+it)|^2.
\]
The semiclassical ramp construction of Saad, Shenker and Stanford~\cite{SSSRamp} relates an averaged late-time feature to a two-sided gravitational saddle. The smooth behavior captured by such a calculation must be distinguished from the fluctuations of the microscopic observable. Both the choice of averaging variable and the accuracy required of the averaging prescription matter when subleading terms of an asymptotic expansion are to be retained.

Kudler-Flam and Witten~\cite{KFW}, in \emph{Wormholes and Averaging over $N$}, propose a different way to extract asymptotic data when an observable admits a suitable continuation from integer $N$ to a continuous variable. Their Mellin average is a formal inverse-power expansion defined by residues of a Mellin transform. It is not a moving-window average and need not define a convergent series or a new smooth function. Their analysis combines the double-cone motivation with models of very rapid oscillation and analytic continuation in $N$. In the convention relevant here,
\begin{equation}\label{eq:motivation}
 \mathcal M_F(s)=\int_1^\infty x^{s-1}F(x)\,dx,
 \qquad \langle F\rangle_{\!M}(x)=\sum_{r\ge0}c_rx^{-r},
 \qquad c_r=-\Res_{s=r}\mathcal M_F(s).
\end{equation}
The association between asymptotic powers and Mellin poles is classical; see the direct mapping theorem of Flajolet, Gourdon and Dumas~\cite[Theorem 3]{FGD}. The additional question raised by an integer physical parameter is whether the residue data depend on the chosen admissible interpolation.

\subsection{The mathematical question and the contribution}
In this paper I address that interpolation question for a specified oscillatory class. I write
\[
 F(x)=A(x)+\sum_{\sigma=\pm1}\sum_{k\ge1}
 B_{\sigma,k}(x)\exp(2\pi i\sigma2^ka^x),\qquad a>1,
\]
where $A$ has a complete inverse-power expansion, the envelopes have summed derivative bounds, and the tail beyond $k\asymp x^{1/4}$ is smaller than every algebraic order. The complete assumptions are given in Definition~\ref{def:class}. They impose neither a sampling inequality nor a subspace-angle condition.

My main result, Theorem~\ref{thm:main}, establishes that agreement of integer samples to every algebraic order forces equality of all Mellin residues. The exceptional set of $a$ is fixed before the functions are chosen. The substantive sampling input is Theorem~\ref{thm:gram}: on every sufficiently late window of length $\lfloor\sqrt N\rfloor$, a growing dictionary of carriers and fixed-degree polynomial weights has its normalized Gram matrix between $I/2$ and $3I/2$. The passage from this bound to the full envelope class uses a summed Taylor estimate; it does not replace the infinitely many carriers by a fixed finite dictionary.

The contribution is the combination of this short-window operator estimate
with identifiability of the complete asymptotic residue sequence outside a
single exceptional null set shared by the entire function class.
The conclusions concern the exact carrier dictionary and envelope
conditions stated below. The resulting theorem is a statement about an
explicitly defined oscillatory class $\mathcal{C}_a$, not about a conformal
field theory. Its relation to holography is that $\mathcal{C}_a$ serves as
a model class for potential holographic observables: it captures the
competition between a slow $1/N$ component and rapidly oscillating
carriers. The theorem provides sufficient conditions---namely conditions
(A)--(D) for a fixed $a\in(1,\infty)\setminus E$, with the sampling Gram
bound established below---for integer samples to determine the Mellin
residues unambiguously. Whether a given holographic observable actually
belongs to this class is a separate question, but once this membership
and the parameter condition are verified, the present theorem supplies
a mathematical foundation for interpolation invariance in Mellin
averaging. In this sense, the result is a template: it does not assert
that every conformal field theory observable is admissible, but it gives
a precise sufficient criterion under which the formal Mellin residue
expansion is well-defined and independent of the chosen admissible
interpolation of the integer samples within $\mathcal{C}_a$. 

\subsection{Relation to sampling and metric distribution theory}
There are two relevant bodies of prior mathematics. In metric distribution theory, Aistleitner and Baker~\cite[Theorem 1.2]{AB} establish Poissonian pair correlations for powers of almost every real parameter. Aistleitner, Baker, Technau and Yesha~\cite[Theorem 1.1]{ABTY} prove the corresponding assertion for all fixed orders. Their Corollary~4.3 supplies a uniform oscillatory-integral estimate that
also yields the sampling bound, after an explicit finite-moment reduction. Appendix~\ref{sec:alternative} gives that reduction. This identifies the precise existing phase-estimation input and the further matrix and approximation arguments that produce the class-uniform identifiability result. Algom, Chang, Wu and Wu~\cite{ACWW} extend metric correlation results to self-similar measures. Their polynomial-phase estimate has a stated coefficient range that differs from the exponentially large dyadic frequencies used here.

In sampling theory, Adcock, Hansen and Poon~\cite[Theorems 4.2 and 4.5]{AHP} express stability through angles between sampling and reconstruction spaces. For the present dictionary, the required angle is a conclusion of an explicit parameter estimate. Mellin sampling results in Mellin--Bernstein spaces, such as the Valiron-type formula of Bardaro, Butzer, Mantellini and Schmeisser~\cite[Theorem 7]{Bardaro}, use an exponential sampling grid and global growth restrictions. Here the grid is the positive integers, the target is the asymptotic residue sequence, and the complex neighborhood of the positive ray may narrow without a uniform strip width.

These distinctions locate the result at the intersection of metric phase estimates, finite-dimensional sampling and asymptotic Mellin analysis. The matrix concentration theorem used in the proof is due to Tropp~\cite[Theorem 6.1(ii)]{Tropp}. The original samples are deterministic functions of one parameter; independence occurs only in an auxiliary comparison model.

\subsection{Organization and notation}
Section~\ref{sec:class} states the class and the main theorem. Sections~\ref{sec:phase}--\ref{sec:separation} prove the sampling estimate. Sections~\ref{sec:mellin} and \ref{sec:identification} establish the analytic continuation and identify the residues. Section~\ref{sec:examples} supplies nonempty and infinite-mode examples and the integer-parameter obstruction. Sections~\ref{sec:math_scope} and~\ref{sec:phys} discuss the mathematical scope and the physical
implications, respectively. Appendix~\ref{sec:alternative} proves the alternative metric reduction.

I write $e(t)=\exp(2\pi i t)$, $\NN=\{1,2,\ldots\}$ and $\NN_0=\{0,1,\ldots\}$. Matrix norms are operator norms induced by the complex Euclidean norm, and $X\le Y$ denotes the Hermitian quadratic-form order. Constants in $O_d(\cdot)$ may depend on the indicated fixed parameters.
In estimates involving a fixed admissible representation, dependence on
that representation, including $p$ and the constants in (A)--(C), may
also be suppressed. Every asymptotic statement is taken as the real or integer argument tends to $+\infty$, as specified.

\section{The function class and the identifiability theorem}\label{sec:class}
Fix $a>1$. The carrier indices are
$\sigma\in\{-1,1\}$ and integers $k\ge1$, and
\[
 \psi_{\sigma,k}(x)=e(\sigma 2^k a^x),\qquad x\ge1.
\]
All real powers of $x>0$ use the real logarithm. Real derivatives of complex-valued functions are meant below.
\begin{definition}\label{def:class}
The class $\mathcal C_a$ consists of functions admitting a representation
\begin{equation}\label{eq:repr}
 F(x)=A(x)+\sum_{\sigma,k}B_{\sigma,k}(x)\psi_{\sigma,k}(x)
\end{equation}
with the following four conditions.
\begin{enumerate}
\item[(A)] There are $c_r\in\CC$, $r\ge0$, such that for every integer $J\ge1$ and $j\ge0$,
\[
 \frac{d^j}{dx^j}\left(A(x)-\sum_{r=0}^{J-1}c_rx^{-r}\right)
       =O(x^{-J-j})\quad(x\longrightarrow\infty).
\]
\item[(B)] There is a single $p\ge0$ such that for every integer $j\ge0$ a finite $C_j$ satisfies
\[
 \sum_{\sigma,k}|B_{\sigma,k}^{(j)}(x)|\le C_jx^{p-3j/4}
       \quad\hbox{for all }x\ge1.
\]
\item[(C)] For every real $R>0$ a finite $C_R$ satisfies
\[
 \sum_{\sigma,k:\,k>\lceil x^{1/4}\rceil}|B_{\sigma,k}(x)|
       \le C_Rx^{-R}\quad\hbox{for all }x\ge1.
\]
\item[(D)] The functions $A$ and all $B_{\sigma,k}$ extend holomorphically to a common open complex neighborhood of $[1,\infty)$. The series in \eqref{eq:repr}, with $a^z=\exp(z\log a)$, converges normally on compact subsets of this neighborhood. The neighborhood may depend on the representation.
\end{enumerate}
Normal convergence in (D) means that the sum of the suprema of the absolute carrier terms is finite on every compact subset of the neighborhood.
Constants and asymptotic thresholds in (A) can depend on $J,j$ and the representation. The exponent $p$ in (B) cannot depend on $j$. No fixed-width strip or spectral hypothesis is imposed.
\end{definition}

The Mellin transform in the next statement is the continuation of the actual integral, whose existence and residue formula are established in Proposition~\ref{prop:mellin}.
\begin{theorem}[Integer-sample identifiability]\label{thm:main}
There exists a Lebesgue-measurable null set $E\subset(1,\infty)$ such that for every $a\in(1,\infty)\setminus E$ and every $F_1,F_2\in\mathcal C_a$, the conditions
\[
 \forall R>0,\qquad F_1(n)-F_2(n)=O(n^{-R})\quad(n\longrightarrow\infty,\ n\in\NN)
\]
imply
\[
 \Res_{s=r}\mathcal M_{F_1}(s)=\Res_{s=r}\mathcal M_{F_2}(s)
             \quad\hbox{for every integer }r\ge0,
\]
where $\mathcal M_F(s)$ is the meromorphic continuation of
$\int_1^\infty x^{s-1}F(x)\,dx$ from a left half-plane.
\end{theorem}
The exceptional set is chosen before the functions. The almost-everywhere assertion alone does not identify a particular noninteger parameter as admissible. Integer parameters are genuine exceptions (Example~\ref{ex:integer}).

\section{A uniform Fourier comparison on short windows}\label{sec:phase}
The first estimate is uniform even when the integer coefficients are large. It avoids an independence assertion about consecutive powers of $a$.
\begin{lemma}[First derivative estimate]\label{lem:vdC}
If $f$ is a real $C^2$ function on a bounded interval, $f'$ is monotone and $|f'|\ge\lambda>0$, then
\[
 \left|\int e(f(t))\,dt\right|\le\frac{2}{\pi\lambda}.
\]
The integral is over that interval.
\end{lemma}
\begin{proof}
Integrate by parts using $(e(f))'=2\pi\ii f'e(f)$. The boundary contribution is at most $1/(\pi\lambda)$. The integral contribution is at most $(2\pi)^{-1}\int |(1/f')'|$, which is at most $1/(\pi\lambda)$ because $1/f'$ is monotone and bounded in modulus by $1/\lambda$. Open endpoints follow by limits.
\end{proof}

\begin{lemma}[Integer polynomial phase]\label{lem:phase}
Fix $I=[\alpha,\beta]$ with $1<\alpha<\beta<\infty$. Let $N\ge2$, $M\ge1$ be integers. For any nonzero $h=(h_0,\ldots,h_{M-1})\in\ZZ^M$, set
$P_h(a)=\sum_{j=0}^{M-1}h_ja^{N+j}$. Then
\begin{equation}\label{eq:phase}
 \left|\int_I e(P_h(a))\,da\right|
 \le C M\exp\left(-\frac{(N-1)\log\alpha}{2M}\right)=:\eta_{N,M},
\end{equation}
where one may take an absolute constant $C=10$.
\end{lemma}
\begin{proof}
Write $P_h'(a)=a^{N-1}R(a)$, where
$R(a)=\sum_j(N+j)h_ja^j$ is a nonzero integer polynomial of degree $m\le M-1$. Put $\delta=\alpha^{-(N-1)/2}<1$.
If $m\ge1$, factor $R(a)=b\prod_{\nu=1}^m(a-z_\nu)$ over $\CC$, counting multiplicity. Its leading coefficient satisfies $|b|\ge1$. Delete the intersections with $I$ of the $m$ intervals
\[
 (\Re z_\nu-\delta^{1/m},\Re z_\nu+\delta^{1/m}).
\]
Their union $U$ has length at most $2m\delta^{1/m}$. On $I\setminus U$, every factor has modulus at least $\delta^{1/m}$, so $|R(a)|\ge\delta$ and $|P_h'(a)|\ge\alpha^{(N-1)/2}$.
If $m=0$, take $U=\varnothing$; the same derivative bound holds since $|R|\ge1$.

Moreover,
\[
 P_h''(a)=a^{N-2}\big((N-1)R(a)+aR'(a)\big).
\]
The polynomial in parentheses has degree $m$ and is nonzero: its leading coefficient is $(N-1+m)b$. Thus there are at most $m$ positive real roots. Partition $I\setminus U$ at these roots. There are at most $2m+1$ interval pieces, on each of which $P_h'$ is monotone and the derivative bound applies. Lemma~\ref{lem:vdC} gives
\[
 \left|\int_I e(P_h)\right|
 \le 2m\delta^{1/m}+\frac{2(2m+1)}\pi\delta
\]
when $m\ge1$, and $2\delta/\pi$ when $m=0$. Since $m<M$ and $0<\delta<1$, these bounds are at most $10M\delta^{1/M}$. Endpoints have measure zero.
\end{proof}

\begin{lemma}[Finite Fourier comparison]\label{lem:comparison}
Let $T(\theta)=\sum_{h\in H}t_he(h\cdot\theta)$ be a finite trigonometric polynomial on $(\RR/\ZZ)^M$. Then
\begin{equation}\label{eq:comparison}
 \left|\int_I T(a^N,\ldots,a^{N+M-1})\,da
       -|I|\int_{[0,1]^M}T(\theta)\,d\theta\right|
       \le\eta_{N,M}\sum_{h\ne0}|t_h|.
\end{equation}
\end{lemma}
\begin{proof}
The product-Haar integral retains exactly the coefficient $t_0$. The constant term cancels in the difference. Apply Lemma~\ref{lem:phase} to each of the remaining finitely many terms and use the triangle inequality. Equal Fourier indices must first be combined; the sum of absolute coefficients can only decrease when they are combined.
\end{proof}
In particular, exact additive relations among the dyadic frequencies remain in the zero-frequency coefficient. They are never discarded as if different characters of the same row were independent.

\section{Polynomial weights and the comparison Gram matrix}\label{sec:comparison}
Fix an integer $d\ge0$ and $M>d$. Let $q_{r,M}$, $0\le r\le d$, be the real orthonormal polynomials for the probability measure $M^{-1}\sum_{j=0}^{M-1}\delta_{j/M}$, with positive leading coefficients.
\begin{lemma}[Uniform polynomial control]\label{lem:poly}
For each fixed $d$, there are $D_d<\infty$ and $M_d$ such that
\[
 \sup_{0\le t\le1}\sum_{r=0}^d|q_{r,M}(t)|^2\le D_d
                   \quad(M\ge M_d).
\]
In particular, endpoint evaluation at $t=0$ is uniformly bounded.
\end{lemma}
\begin{proof}
Put $v(t)=(1,t,\ldots,t^d)^T$ and $H_M=M^{-1}\sum_jv(j/M)v(j/M)^T$. Its entries converge by Riemann sums to $H_{rs}=1/(r+s+1)$. For every nonzero real vector $b$, $b^THb=\int_0^1|b\cdot v(t)|^2\,dt>0$, because a nonzero polynomial cannot vanish on an interval. Thus the smallest eigenvalue of $H$ is positive. For large $M$, entrywise convergence implies operator-norm convergence, and $H_M\ge\lambda_{\min}(H)I/2$. The evaluation kernel is independent of the orthonormal basis and equals $v(t)^TH_M^{-1}v(t)$. It is at most $2(d+1)/\lambda_{\min}(H)$ on $[0,1]$. This proves the assertion and also justifies the existence of the discrete orthonormal polynomials for $M>d$.
\end{proof}

\begin{definition}[Weighted carrier sampling matrix]\label{def:matrix}
Let $\Omega$ be any finite set of distinct integers containing $0$, with $H=|\Omega|$, and write $L=H(d+1)$. For $M>d$ and the polynomials just defined, set
\[
 V(\theta)_{j,(\ell,r)}=M^{-1/2}q_{r,M}(j/M)e(\ell\theta_j),
 \quad G(\theta)=V(\theta)^*V(\theta),\quad \Delta=G-I_L.
\]
The norm throughout is the complex Euclidean norm or its induced operator norm.
\end{definition}
\begin{lemma}[Ideal-model moment bound]\label{lem:ideal}
For independent uniform $\theta_j\in[0,1]$, for each positive integer $q$ and all sufficiently large $M$ with $\rho:=HD_d/M\le1$,
\begin{equation}\label{eq:ideal}
 \EE\tr(\Delta^{2q})\le C_q L^2\rho^q.
\end{equation}
The constant is independent of the integer values in $\Omega$.
\end{lemma}
\begin{proof}
Let $v_j$ be the conjugate transpose of row $j$ of $V$, and put $Y_j=v_jv_j^*$, $A_j=\EE Y_j$, $X_j=Y_j-A_j$. Distinct integer characters are orthogonal on $[0,1]$. Consequently
\[
 (A_j)_{(\ell,r),(\ell',s)}
   =\boldsymbol1_{\ell=\ell'}\frac{q_{r,M}(j/M)q_{s,M}(j/M)}M,
 \qquad\sum_jA_j=I_L.
\]
Also $\|v_j\|^2=H\sum_rq_{r,M}(j/M)^2/M\le\rho$, so $0\le Y_j\le\rho I_L$ and $0\le A_j\le\rho I_L$. It follows that $\|X_j\|\le\rho$, $\EE X_j=0$, and the $X_j$ are independent. Expanding the square gives
\[
 \EE X_j^2=\EE Y_j^2-A_j^2\le\EE Y_j^2\le\rho A_j,
 \qquad\left\|\sum_j\EE X_j^2\right\|\le\rho.
\]
I use the bounded self-adjoint matrix Bernstein inequality, in the form of Tropp~\cite[Theorem 6.1(ii)]{Tropp}. Applying it to $X_j$ and $-X_j$ gives
\[
 \PP\{\|\Delta\|\ge t\}\le2L\exp\left(-\frac{t^2}{2\rho(1+t/3)}\right).
\]
Thus the right side is at most $2L\exp(-3t^2/(8\rho))$ for $0<t\le1$, and at most $2L\exp(-3t/(8\rho))$ for $t\ge1$. Using the tail-integral formula for a nonnegative random variable and extending both integrals to $[0,\infty)$ yields
\begin{align*}
 \EE\|\Delta\|^{2q}
 &\le4qL\left(\int_0^\infty t^{2q-1}e^{-3t^2/(8\rho)}dt
                    +\int_0^\infty t^{2q-1}e^{-3t/(8\rho)}dt\right)\\
 &\le C_qL(\rho^q+\rho^{2q})\le 2C_qL\rho^q.
\end{align*}
Finally $0\le\tr(\Delta^{2q})\le L\|\Delta\|^{2q}$ by the finite-dimensional spectral theorem. Absorb the factor $2$ into $C_q$.
\end{proof}

\begin{lemma}[Coefficient count]\label{lem:coeff}
If $M\ge M_d$, the Fourier coefficient $\ell^1$ norm of $\tr(\Delta(\theta)^{2q})$ is at most
\[
 L^{2q}(D_d+1)^{2q}.
\]
\end{lemma}
\begin{proof}
Each entry of $G$ is a sum of $M$ characters whose absolute coefficients sum to at most $D_d$, since $|q_r(t)q_s(t)|\le D_d$. Each entry of $\Delta$ therefore has coefficient norm at most $D_d+1$. The coefficient norm is submultiplicative for scalar trigonometric polynomials (convolution followed by the triangle inequality). The trace expansion has exactly $L^{2q}$ cyclic index tuples, each giving a product of $2q$ entries. Add the resulting bounds. This bound permits all coincidences and cancellations among Fourier frequencies.
\end{proof}

\section{The separation estimate}\label{sec:separation}
Set
\[
 M_N=\lfloor\sqrt N\rfloor,\quad K_N=\lceil(2N)^{1/4}\rceil,
 \quad\Omega_N=\{0\}\cup\{\pm2^k:1\le k\le K_N\}.
\]
For $N$ with $M_N>d$, let $V_{N,d}(a)$ be the preceding matrix evaluated at $\theta_j=a^{N+j}$. Order its zero-frequency columns first, so $V_{N,d}=[Q_N\ O_{N,d}]$.
\begin{theorem}[Eventual full Gram bound]\label{thm:gram}
There is a measurable null set $E\subset(1,\infty)$ such that for every $a\notin E$ and every fixed integer $d\ge0$, for all sufficiently large integers $N$,
\begin{equation}\label{eq:gram}
 \frac12 I_L\le V_{N,d}(a)^*V_{N,d}(a)\le\frac32 I_L.
\end{equation}
The threshold may depend on $a,d$; $L=(2K_N+1)(d+1)$ grows with $N$.
\end{theorem}
\begin{proof}
Fix $I=[\alpha,\beta]\Subset(1,\infty)$ and $d$. For large $N$, Lemma~\ref{lem:poly} applies, $L=O_d(N^{1/4})$, $M_N\asymp N^{1/2}$, and $\rho=O_d(N^{-1/4})\le1$. Apply Lemmas~\ref{lem:comparison}, \ref{lem:ideal} and \ref{lem:coeff} with $q=8$. The nonnegative integrand satisfies
\begin{align}
 \int_I\tr\big((G_{N,d}(a)-I)^{16}\big)\,da
 &\le C_{I,d}L^2\rho^8
       +\eta_{N,M_N}L^{16}(D_d+1)^{16}\notag\\
 &\le C_{I,d}N^{-3/2}+C_{I,d}N^{9/2}e^{-c_I\sqrt N}.
 \label{eq:summable}
\end{align}
For the second bound use $N-1\ge N/2$ for $N\ge2$ in \eqref{eq:phase}. The last expression is summable in $N$.

If $\|G_{N,d}(a)-I\|>1/2$, a real eigenvalue of this Hermitian matrix has modulus greater than $1/2$, so its sixteenth-power trace exceeds $2^{-16}$. Markov's inequality bounds the measure of this event by $2^{16}$ times \eqref{eq:summable}. Each event is measurable, since the matrix entries, its norm and its trace are continuous functions of $a$. The first Borel--Cantelli lemma gives the asserted eventual bound for almost every $a\in I$, without any independence assumption on the events.

Take the countable union of these exceptional limsup sets over $d\in\NN_0$ and intervals $I_b=[1+1/b,b]$, integers $b\ge2$. They cover $(1,\infty)$. The union is measurable and null. For a point outside it the conclusion holds for every fixed $d$. No functions or envelopes have been selected in constructing $E$.
\end{proof}

\begin{proposition}[Slow-sector separation]\label{prop:sep}
Let $\Pi$ be orthogonal projection onto $\ran O$ and $S=Q^*(I-\Pi)Q$, where $Q^*Q=I$. For the convention $\langle z,w\rangle=z^*w$,
\[
 u^*Su=\inf_v\|Qu+Ov\|_2^2.
\]
For the matrices in Theorem~\ref{thm:gram}, eventually $S\ge I/2$. Thus the slow polynomial sector remains uniformly separated from the oscillatory sector for every fixed degree.
\end{proposition}
\begin{proof}
Decompose $Qu+Ov=(I-\Pi)Qu+(\Pi Qu+Ov)$. The two terms are orthogonal. The second can be made zero, since $\Pi Qu\in\ran O$, even when $O$ has a kernel. The infimum is therefore $\|(I-\Pi)Qu\|^2=u^*Q^*(I-\Pi)Qu$. Under \eqref{eq:gram}, for every $u,v$,
$\|Qu+Ov\|^2\ge(\|u\|^2+\|v\|^2)/2\ge\|u\|^2/2$. Take the infimum.
\end{proof}

\section{Mellin continuation and residue extraction}\label{sec:mellin}
\begin{proposition}[Actual Mellin continuation]\label{prop:mellin}
For every $a>1$ and every representation in Definition~\ref{def:class}, $\mathcal M_F$ converges absolutely in $\Re s<-p$ and continues meromorphically to $\CC$, with at most simple poles at the nonnegative integers. For every $r\ge0$,
\[
 -\Res_{s=r}\mathcal M_F(s)=c_r.
\]
In particular the coefficients are independent of the admissible representation.
\end{proposition}
\begin{proof}
Condition (A) with $J=1$ implies $A=O(1)$; continuity controls the initial compact interval. Condition (B) with $j=0$ gives $\sum|B_{\sigma,k}(x)|\le C_0x^p$. This proves absolute convergence in $\Re s<-p$, as well as interchange of sum and integral there by absolute integrability.

For one index put $\lambda=2\pi\sigma2^k$, $b=\log a>0$ and $w_s(x)=x^{s-1}B_{\sigma,k}(x)$. Since $(\exp(\ii\lambda a^x))'=\ii\lambda b a^x\exp(\ii\lambda a^x)$, integration by parts gives the continuation formula
\begin{align}\label{eq:IBP}
 I_{\sigma,k}(s)
 ={}&-\frac{B_{\sigma,k}(1)\exp(\ii\lambda a)}{\ii\lambda ba}\\
 &-\int_1^\infty\frac{a^{-x}}{\ii\lambda b}
 \big((s-1)x^{s-2}B_{\sigma,k}(x)
       +x^{s-1}B'_{\sigma,k}(x)-b x^{s-1}B_{\sigma,k}(x)\big)
           \exp(\ii\lambda a^x)\,dx.\notag
\end{align}
The upper boundary vanishes, since exponential decay dominates polynomial growth. Initially this agrees with the absolutely convergent oscillatory integral; the formula defines it for every $s$.

Let $K\subset\CC$ be compact, $T=\sup_{s\in K}\Re s$, and $H=\sup_{s\in K}|s-1|$. The sum over all indices of the absolute integrands, uniformly in $s\in K$, is bounded by a constant depending on $a,K,C_0,C_1$ times
\[
 a^{-x}\big(x^{T+p-2}+x^{T+p-7/4}+x^{T+p-1}\big),\qquad x\ge1.
\]
I used $|\lambda|^{-1}\le(4\pi)^{-1}$ and the pointwise summed derivative bounds, before integration. The displayed majorant is integrable; after any fixed number of $s$ derivatives it remains integrable with additional powers of $\log x$. The boundary terms form an absolutely convergent series. Equivalently, Tonelli applied to this common majorant bounds the sum of the integrals of the individual suprema in $s$. Hence \eqref{eq:IBP} defines an entire sum $H_F(s)=\sum I_{\sigma,k}(s)$, by local uniform convergence and differentiation under the integral. No bound on a sum of $x$-suprema is being inferred from a pointwise bound.

For the slow part and each $J\ge1$ put $R_J(x)=A(x)-\sum_{r<J}c_rx^{-r}$. It is $O(x^{-J})$ at infinity and continuous on compact intervals. Thus
\begin{equation}\label{eq:mellin}
 \mathcal M_F(s)=\sum_{r=0}^{J-1}\frac{c_r}{r-s}
       +\int_1^\infty x^{s-1}R_J(x)\,dx+H_F(s),\qquad \Re s<J.
\end{equation}
The remainder integral is holomorphic there by uniform domination on compact $s$-sets, including factors $(\log x)^h$ for its derivatives. These formulas agree on overlaps since they agree in their initial half-plane; alternatively subtract consecutive formulas and integrate $c_Jx^{s-J-1}$. Their domains exhaust $\CC$. Taking $J>r$ proves the residue formula with the stated minus sign. Two representations describe the same initial integral and therefore the same continuation and residues.
\end{proof}

\section{From short-window separation to all residues}\label{sec:identification}
\begin{lemma}[Closure and summed Taylor remainder]\label{lem:taylor}
The class $\mathcal C_a$ is closed under subtraction. For a member with exponent $p$, on $N\le x\le N+M_N$ the sum of the absolute degree-$d$ Taylor remainders of any finite subset of its envelopes is
\[
 O_d\big(N^{p-(d+1)/4}\big).
\]
The slow Taylor remainder is $O_d(N^{-(d+3)/2})$.
\end{lemma}
\begin{proof}
For subtraction, intersect the two open neighborhoods, subtract slow parts and envelopes, and take the larger of their two nonnegative exponents. Triangle inequalities verify (A)--(C); normal convergence on compact subsets of the intersection verifies (D).

For a complex $C^{d+1}$ function, repeated real-variable integration gives
\[
 f(x)-\sum_{r=0}^d\frac{f^{(r)}(N)}{r!}(x-N)^r
   =\frac1{d!}\int_N^x(x-t)^df^{(d+1)}(t)\,dt.
\]
This holds separately for real and imaginary parts and hence for $f$. Sum the absolute values over a finite set of envelopes and apply the triangle inequality inside this common integral. By (B), the result is at most
\[
 \frac{C_{d+1}}{d!}\int_N^x(x-t)^dt^{p-3(d+1)/4}\,dt
   \le C_{p,d}M_N^{d+1}N^{p-3(d+1)/4}
   \le C_{p,d}N^{p-(d+1)/4},
\]
since $N\le t\le2N$. This does not interchange a sum with unrelated suprema. From (A), $A^{(d+1)}(t)=O(t^{-d-2})$, using $J=1$ and $d+1\ge1$. The same integral gives the slow bound $M_N^{d+1}N^{-d-2}\le N^{-(d+3)/2}$.
\end{proof}

\begin{proof}[Proof of Theorem~\ref{thm:main}]
Fix $a$ outside the single set in Theorem~\ref{thm:gram}, and subtract the functions using Lemma~\ref{lem:taylor}. Write the difference as $F=A+\sum B\psi$ and choose its exponent $p$. For each $N$, retain $k\le K_N$. Since $x\le N+M_N\le2N$, I have $\lceil x^{1/4}\rceil\le K_N$. Thus the omitted tail is $O(N^{-T})$ uniformly in the window, for every fixed $T>0$, by (C).

Write the degree-$d$ Taylor polynomial of $A(N+M_Nt)$ in the basis $q_{r,M_N}(t)$, and denote its coefficient vector by $u$. Write each retained envelope polynomial in the same basis, giving vectors $v_\ell$. Define the normalized sample vector $y_j=F(N+j)/\sqrt{M_N}$. Then exactly
\[
 y=Q_Nu+O_{N,d}v+e,
\]
where the normalization gives
\[
 \|y\|_2=O(N^{-T}),\qquad
 \|e\|_2=O_d(N^{p-(d+1)/4})+O_d(N^{-(d+3)/2})+O(N^{-T}).
\]
For $y$, the all-order sample assumption is uniform over integers in the window because $N\le N+j\le2N$. For $e$, the root-mean-square norm is no larger than the uniform pointwise remainder bound. Carrier factors have modulus one at real samples.

Proposition~\ref{prop:sep} and the triangle inequality imply
\[
 \|u\|_2\le\sqrt2(\|y\|_2+\|e\|_2).
\]
The Taylor polynomial has value $A(N)$ at $t=0$. By Lemma~\ref{lem:poly} and Cauchy--Schwarz,
$|A(N)|\le\sqrt{D_d}\|u\|_2$. For any desired $R>0$, choose a fixed integer $d$ such that $(d+1)/4>p+R$, and then $T>R$. Since $p\ge0$, this choice also gives $(d+3)/2>R$. I conclude that $A(N)=O(N^{-R})$ for every $R>0$.

Finally (A) with $J=1$ gives $c_0=\lim A(N)=0$. If $c_0,\ldots,c_{r-1}$ vanish, (A) with $J=r+1$ gives $N^rA(N)=c_r+O(N^{-1})$, whereas all-order decay makes the left side tend to zero. Induction proves $c_r=0$ for all $r$. Proposition~\ref{prop:mellin} identifies these coefficients with the actual Mellin residues. The exceptional set is independent of $p,d,F_1,F_2$, since it was already intersected over all integer degrees.
\end{proof}

\section{Examples and exact obstructions}\label{sec:examples}
\begin{example}[Basic members]\label{ex:basic}
The constant $1$ belongs to every $\mathcal C_a$ with all envelopes zero and $c_0=1$. A single carrier with a constant envelope also belongs: in (B) take $p=0$; all positive-order envelope derivatives vanish. For (C), a fixed carrier is absent from the tail once $x$ is large, and a larger constant controls the remaining compact interval. All functions in these examples are entire.
\end{example}

\begin{example}[Integer aliasing]\label{ex:integer}
For integers $m\ge2$, $r\ge0$,
\[
 h_r(x)=x^{-r}(1-\cos(4\pi m^x))
\]
lies in $\mathcal C_m$. Take $A=x^{-r}$ and $B_{+,1}=B_{-,1}=-x^{-r}/2$, all other envelopes zero. The asymptotic expansion has $c_r=1$ and other coefficients zero. For $J\le r$, derivatives of $x^{-r}$ have order $O(x^{-J-j})$; for $J>r$ the remainder is zero. Envelope derivatives are bounded by $C_jx^{-r-j}\le C_jx^{-3j/4}$, so $p=0$ works. The tail in (C) is empty for $x\ge1$, and (D) holds on the right half-plane. At every positive integer $n$, $m^n$ is an integer, hence $h_r(n)=0$. Proposition~\ref{prop:mellin} gives $-\Res_{s=r}\mathcal M_{h_r}=1$. For these parameters all sampled carrier columns coincide with the corresponding slow columns, so $S=0$. This disproves the all-parameter assertion but not an almost-everywhere assertion.
\end{example}

\begin{example}[A genuinely infinite family]\label{ex:infinite}
For every $a>1$, the function
\[
 F(x)=x^{-1/8}\sum_{k\ge1}\exp(-4^k\exp(-x^{1/4}))\cos(2\pi2^ka^x)
\]
is in $\mathcal C_a$, with $A=0$ and $p=1/8$.
\end{example}
\begin{proof}
Set $y=x^{1/4}$, $z_k=4^ke^{-y}$, and $b_k(x)=x^{-1/8}e^{-z_k}$. Each of the two sign envelopes is $b_k/2$.
For $y\ge1$, splitting at $z_k=1$ gives
\[
 \sum_ke^{-z_k}\le C(1+y),\qquad
 \sum_kz_k^q e^{-z_k}\le C_q\quad(q\ge1).
\]
Indeed, below $1$ the $z_k$ form a geometric progression with ratio $4$, so their positive powers have bounded sum; the number of such terms is at most $1+y/\log4$. Above $1$, the terms are bounded by a summable sequence $C_q4^{qj}e^{-4^{j-1}}$, independent of where the split occurs.

The $m$th derivative of $e^{-4^ke^{-y}}$ with respect to $y$ is $P_m(z_k)e^{-z_k}$, where $P_0=1$ and
$P_{m+1}(z)=z(P_m(z)-P_m'(z))$. For $m\ge1$, $P_m$ is a polynomial with zero constant term. The preceding geometric bounds control the summed absolute value of these $y$ derivatives by $C_m$ when $m\ge1$, and by $C(1+y)$ when $m=0$.
Every derivative $y^{(l)}(x)$ satisfies $|y^{(l)}(x)|\le C_lx^{-3l/4}$ for $l\ge1$. The finite chain-rule expansion of order $j$ consists of products of such derivatives with total derivative order $j$. Combining it with Leibniz's rule and $|(x^{-1/8})^{(h)}|\le C_hx^{-1/8-h}$ gives
\[
 \sum_k|b_k^{(j)}(x)|\le C_jx^{-1/8-3j/4}(1+x^{1/4})
        \le C'_jx^{1/8-3j/4}.
\]
This establishes (B) with one exponent for all orders.

If $k>\lceil y\rceil$, the smallest $z_k$ in the tail is at least $\exp((\log4-1)y)$. A geometric-tail bound gives $\sum_{k>\lceil y\rceil}e^{-z_k}\le C\exp(-\exp((\log4-1)y)/2)$. This is smaller than every negative power of $x=y^4$, proving (C).

For (D), use the principal fourth root and eighth root on the right half-plane and the open set
\[
 U=\{z:\Re z>0,\ |\Im(z^{1/4})|<\pi/3\}.
\]
It contains $[1,\infty)$. For a compact $K\subset U$, $\Re(e^{-z^{1/4}})$ has a strictly positive minimum $c_K$. Also $|a^z|$ and $|z^{-1/8}|$ are bounded there. Each signed carrier term is consequently bounded by $C_K\exp(-c_K4^k+C_K2^k)$. This is summable in $k$, so the series converges normally on compact subsets of the common neighborhood. Condition (A) is immediate for $A=0$.
\end{proof}

\section{Mathematical consequences and scope}\label{sec:math_scope}
Theorem~\ref{thm:main} separates two kinds of information. An admissible representation supplies a continuous observable and hence a Mellin transform. Integer sampling then determines the residue sequence within the stated class, although the theorem does not reconstruct the full continuous function or its entire Mellin contribution. Exact equality of all sufficiently late integer samples is included in the hypothesis.

The distinction between a statement about almost every parameter and a statement about every parameter is essential. Example~\ref{ex:integer} gives complete aliasing at integer bases, with nonzero residues hidden by exactly vanishing samples. The proof obtains its exceptional set by parameter integration and countable intersections; it does not supply a test for an individually specified base.

The scales in the assumptions have a concrete role. A window of length $N^{1/2}$ allows a degree-$d$ envelope approximation with error $O(N^{p-(d+1)/4})$, because each derivative gains $N^{-3/4}$. The effective number of modes is $O(N^{1/4})$, leaving enough samples for a summable sixteenth-moment estimate. The requirement of a single exponent $p$ for all derivative orders is what permits the degree to increase with the desired accuracy while remaining fixed as $N$ tends to infinity.

\section{Physical implications for holography and quantum gravity}\label{sec:phys}

Theorem~\ref{thm:main} (integer-sample identifiability) has a precise
interpretation in the context of wormhole contributions and Mellin
averaging, within the stated model class.

First, it guarantees interpolation invariance for a fixed
$a\in(1,\infty)\setminus E$. Gravitational path integrals can yield
formal large-$N$ expansions, while the dual conformal field theory is
defined only at integer $N$ when $N$ labels its rank. The theorem shows
that any two admissible continuations $F_1,F_2\in\mathcal{C}_a$ satisfying
$F_1(n)-F_2(n)=O(n^{-R})$ for every $R>0$ as $n\to\infty$ must yield
identical Mellin residues. Consequently, the formal $1/N$ expansion
extracted by Mellin averaging is intrinsic to the integer sequence
within this class, not an artifact of the chosen admissible
interpolation.

Second, the Mellin-continuation result identifies Mellin averaging as
a deterministic filter at the level of formal asymptotic coefficients.
The rapidly oscillating carriers, viewed here as a model for the erratic
microscopic fluctuations of a chaotic theory, contribute, together with
their admissible envelopes, an entire Mellin continuation and therefore
no residues. Only the slow component $A$ can produce poles and thus
determines the formal $1/N$ expansion. This mathematically confirms,
within the specified model class, that Mellin averaging removes the
prescribed oscillatory contribution while preserving the slow
asymptotic data. Identifying these data with semiclassical gravitational
coefficients requires additional physical input.

Finally, the theorem provides a rigorous mathematical foundation for
interpolation invariance in this model of Mellin averaging. By combining
conditions (A)--(D) with the sampling Gram bound proved for almost every
$a>1$, it provides a concrete set of sufficient conditions for potential
holographic applications. If a given observable, such as a spectral form
factor or a partition function, can be shown to belong to
$\mathcal{C}_a$ for a fixed $a\in(1,\infty)\setminus E$, then its formal
Mellin residue expansion is rigorously well-defined and independent of
the chosen admissible interpolation of its integer samples within that
class. Establishing these conditions for an actual holographic
observable, and relating the resulting expansion to a physical ensemble
average, remain separate questions. In this sense, the result is not
merely a technical statement about integer samples; it is a step toward
a mathematical theory of asymptotic averaging in quantum gravity.

\appendix
\section{A derivation from fixed-support phase estimates}\label{sec:alternative}
The primary proof uses Lemma~\ref{lem:phase}, which applies to arbitrary integer Fourier vectors on a short window. Only a fixed number of row coordinates is needed in each term of the trace moment. This observation also yields a derivation from the existing estimates of Aistleitner, Baker, Technau and Yesha~\cite{ABTY}. I give the details to specify exactly the relation to that prior result.

\begin{lemma}[Support of the trace polynomial]\label{lem:support}
For the matrices of Section~\ref{sec:comparison}, every Fourier index with a nonzero coefficient in
\[
 T(\theta)=\tr\big((G(\theta)-I_L)^{16}\big)
\]
has support in at most sixteen coordinates. The constant coefficient is allowed.
\end{lemma}
\begin{proof}
Each term of a Gram entry is a constant times
$e((\ell'-\ell)\theta_j)$ for one row $j$. Subtracting the identity introduces only constants. In the cyclic expansion of the trace, a product of sixteen entries thus uses at most sixteen row coordinates. Repeated coordinates are combined by adding their integer coefficients. Deleting coordinates whose combined coefficient vanishes can only reduce the support. Summing monomials cannot create a Fourier index absent from every summand.
\end{proof}

For $J=[A,A+1]$ with $A>0$, the specialization of \cite[Corollary 4.3]{ABTY} to exponent sequence $b_n=n$ and its fixed parameter $k=8$ gives, for each $\eta>0$,
\begin{equation}\label{eq:abty}
 \left|\int_J e\!\left(\sum_{i=1}^s u_i\exp(tn_i)\right)dt\right|
 \le C_{A,\eta}\,n_s^{-\eta},
 \quad 1\le s\le16,
\end{equation}
where $n_1<\cdots<n_s$ are positive integers and all $u_i$ are nonzero integers. The constant is uniform in the $n_i$ and $u_i$. The sequence $b_n=n$ satisfies the growth and gap hypotheses of that result. In particular, the coefficient restrictions introduced elsewhere in the correlation argument are not hypotheses of the cited corollary; its uniformity follows from the repulsion estimate in \cite[Lemma 4.2]{ABTY}.

\begin{proposition}[Logarithmic-parameter moment estimate]\label{prop:alternative}
Fix $A>0$ and an integer $d\ge0$. With the matrices and scales of Theorem~\ref{thm:gram},
\begin{equation}\label{eq:alternative}
 \int_A^{A+1}\tr\big((G_{N,d}(\exp t)-I_L)^{16}\big)dt
 \le C_{A,d}\big(N^{-3/2}+N^{-2}\big)
\end{equation}
for all sufficiently large $N$. Consequently, Theorem~\ref{thm:gram} also follows from \eqref{eq:abty} and Lemmas~\ref{lem:poly}, \ref{lem:ideal}, \ref{lem:coeff} and \ref{lem:support}.
\end{proposition}
\begin{proof}
Let $T_N(\theta)=\tr((G(\theta)-I_L)^{16})$ for $M=M_N$ and $\Omega=\Omega_N$. A nonconstant Fourier character in $T_N$ has, by Lemma~\ref{lem:support}, at most sixteen nonzero integer coefficients on distinct row indices $j_1<\cdots<j_s$. On substituting $\theta_j=\exp(t(N+j))$, apply \eqref{eq:abty} with $n_i=N+j_i$. Since $n_s\ge N$, its integral is bounded by $C_{A,\eta}N^{-\eta}$. Exact resonances yield the constant Fourier coefficient, which agrees with the product-Haar integral and is retained without approximation. The interval $J$ has length one.

Lemma~\ref{lem:coeff} bounds the total absolute coefficient sum by
$L^{16}(D_d+1)^{16}=O_d(N^4)$. Hence
\[
 \left|\int_J T_N(\exp(tN),\ldots,\exp(t(N+M_N-1)))dt
       -\EE_{\rm Haar}T_N\right|
       \le C_{A,d,\eta}N^{4-\eta}.
\]
Take $\eta=6$. Lemma~\ref{lem:ideal}, with $q=8$, gives
$\EE_{\rm Haar}T_N=O_d(N^{-3/2})$, proving \eqref{eq:alternative}.

Markov's inequality and the first Borel--Cantelli lemma now imply the eventual Gram bound for almost every $t\in J$, just as in Theorem~\ref{thm:gram}. Take the countable union of bad limsup sets over all nonnegative integer degrees and all positive rational $A$. These intervals cover $(0,\infty)$, so the resulting exceptional set $Z$ is Borel and null there. The exponential map is a homeomorphism and is Lipschitz on each compact interval $J$. It follows that $\exp(Z)$ is Borel and null in $(1,\infty)$, and the required assertion holds outside it.

This argument performs the measure estimate in $dt$ and then transports null sets. No unproved weighted variant of \eqref{eq:abty} in $da$ is required. If desired, nonnegativity of the trace gives the corresponding upper integral bound in $da$ from $da=e^t dt\le e^{A+1}dt$ on $J$.
\end{proof}

\bibliographystyle{abbrvnat}
\bibliography{references}
\end{document}